\documentclass[12pt]{iopart}
\expandafter\let\csname equation*\endcsname\relax
\expandafter\let\csname endequation*\endcsname\relax

\usepackage{amsmath,amsfonts,amssymb,amsthm}
\usepackage{graphicx}
\usepackage{color}
\usepackage{hyperref}

\def\softd{{\leavevmode\setbox1=\hbox{d}%
\hbox to 1.05\wd1{d\kern-0.4ex{\char039}\hss}}}%cstocs
\def\softt{{\leavevmode\setbox1=\hbox{t}%
\hbox to \wd1{t\kern-0.6ex{\char039}\hss}}}%cstocs

\newcommand{\R}{\mathbb{R}}

\newtheorem{theorem}{Theorem}[section]
\newtheorem{lemma}{Lemma}[section]

\newtheorem{remark}{Remark}[section]

\begin{document}

\title[Schr\"odinger operators exhibiting a spectral transition]
{Spectral analysis of a class of Schr\"odinger operators exhibiting a parameter-dependent spectral transition}

\author{Diana Barseghyan}
\address{Department of Mathematics, University of Ostrava, 30. dubna 22, 70103 Ostrava, Czech Republic}
\address{Nuclear Physics Institute, Academy of Sciences of the Czech Republic,
Hlavn\'{i} 130, 25068 \v{R}e\v{z} near Prague, Czech Republic}
\ead{diana.barseghyan@osu.cz}

\author{Pavel Exner}
\address{Nuclear Physics Institute, Academy of Sciences of the Czech Republic,
Hlavn\'{i} 130, 25068 \v{R}e\v{z} near Prague, Czech Republic}
\address{Doppler Institute, Czech Technical University, B\v{r}ehov\'{a} 7, 11519 Prague, Czech Republic}
\ead{exner@ujf.cas.cz}

\author{Andrii Khrabustovskyi}
\address{Institute of Analysis, Karlsruhe
Institute of Technology, Englerstr. 2, 76131 Karlsruhe, Germany}
\ead{andrii.khrabustovskyi@kit.edu}

\author{Milo\v{s} Tater}
\address{Nuclear Physics Institute, Academy of Sciences of the Czech Republic,
Hlavn\'{i} 130, 25068 \v{R}e\v{z} near Prague, Czech Republic}
\address{Doppler Institute, Czech Technical University, B\v{r}ehov\'{a} 7, 11519 Prague, Czech Republic}
\ead{tater@ujf.cas.cz}

\begin{abstract}
We analyze two-dimensional Schr\"odinger operators with the potential $|xy|^p - \lambda (x^2+y^2)^{p/(p+2)}$ where $p\ge 1$ and $\lambda\ge 0$, which exhibit an abrupt change of its spectral properties at a critical value of the coupling constant $\lambda$. We show that in the supercritical case the spectrum covers the whole real axis. In contrast, for $\lambda$ below the critical value the spectrum is purely discrete and we establish a Lieb-Thirring-type bound on its moments. In the critical case the essential spectrum covers the positive halfline while the negative spectrum can be only discrete, we demonstrate numerically the existence of a ground state eigenvalue.

\vspace{2pc}
\end{abstract}

\submitto{J. Phys. A.: Math. Theor.}

\maketitle

%%%%%%%%%%%%%%%%%%%%%%%%%%%%%%%%%%%%%%%
\section{Introduction} \label{s: intro}

One of the problems which attracted attention recently concerns Schr\"odinger operators with potentials dependent on a parameter which exhibit a sudden spectral transition when the value of the parameter passes a critical value. The potential is typically unbounded from below and has narrow channels through which the particle can `escape to infinity' in the supercritical situation. Possibly the best know example of this type is the so-called Smilansky model \cite{Sm04, So04, NS06, ES05, Gu11} and its regular version \cite{BE14}. Another example, which will be the main subject of this paper, is a modification of the well-known potential $|xy|^p$ in $\R^2$ obtained by adding a rotationally symmetric negative component which becomes stronger with the growing radius, see \eqref{operator} below. Recall that without the negative component this potential and its modifications serves to demonstrate the possibility of a purely discrete spectrum in the situation when the classically a!
 llowed volume of the phase space is infinite \cite{Si83, GW11, CR15}.

The mechanism of the spectral transition comes from the balance between the negative part of the potential and the positive contribution to the energy coming from the transverse confinement to a channel narrowing towards infinity. This means that the behavior of the two potential components at large distances from the origin must be properly correlated. In our case this is achieved by considering the following class of operators,
 % ------------- %
 \begin{equation} \label{operator}
 L_p(\lambda)\,:\; L_p(\lambda)\psi= -\Delta\psi + \left( |xy|^p - \lambda (x^2+y^2)^{p/(p+2)} \right)\psi\,, \quad p\ge 1\,,
 \end{equation}
 % ------------- %
on $L^2(\R^2)$, where $(x,y)$ in $\R^2$ are the Cartesian coordinates $(x,y)$ in $\R^2$ and the non-negative parameter $\lambda$ in the second term of the potential will serve to control the transition. Note that $\frac{2p}{p+2}<2$, and consequently, the operator (\ref{operator}) is essentially self-adjoint on $C_0^\infty(\R^2)$ by Faris-Lavine theorem -- cf. \cite{RS75}, Thms.~X.28 and  X.38; in the following the symbol $L_p(\lambda)$ will always mean its closure.

We have found already some properties of these operators in \cite{EB12}, our aim here is to present a deeper spectral analysis. To describe what is know we need the (an)harmonic oscillator Hamiltonian on line,
 % ------------- %
\begin{equation}\label{tildeH}
H_p : H_p u = -u^{\prime\prime} + |t|^p u
\end{equation}
 % ------------- %
on $L^2(\mathbb{R})$ with the standard domain, more exactly, its principal eigenvalue $\gamma_p$; since the potential has a mirror symmetry and the ground state is even, we can equivalently consider the `cut' (an)harmonic oscillator on $L^2(\mathbb{R}_+)$ with Neumann condition at $t = 0$. The eigenvalue is known exactly for $p = 2$ where it equals one as well as for $p\to\infty$ where the potential becomes an infinitely deep rectangular well of width two and $\gamma_\infty = \frac{1}{4}\pi^2$. It is easy to see that the function $p\mapsto \gamma_p$ is continuous and positive on the interval $[1,\infty)$; a numerical solution shows that it reaches the minimum value $\gamma_p \approx 0.998995$ at $p \approx 1.788$.

In the paper \cite{EB12} we have shown that the spectral transition occurs at the value $\lambda_\mathrm{crit}=\gamma_p\:$: the spectrum of $L_p(\lambda)$ is purely discrete and below bounded for $\lambda< \lambda_\mathrm{crit}$, remaining below bounded for $\lambda= \lambda_\mathrm{crit}$, while for $\lambda>\lambda_\mathrm{crit}$ it becomes unbounded from below. We have also derived there crude bounds on eigenvalue sums in the subcritical case. In the present work we are going to establish first that for $\lambda>\lambda_{\mathrm{crit}}$ the spectrum of $L_p(\lambda)$ covers the whole real line. Next we shall analyze in more detail the critical case, $\lambda=\lambda_{\mathrm{crit}}$, showing that one has
 % ------------- %
$$
\sigma_{\mathrm{ess}}(L_p(\lambda_{\mathrm{crit}}))=[0, \infty)\,.
$$
 % ------------- %
The question of existence of a negative discrete spectrum is addressed numerically. We show that there a range of values of $p$ for which the critical operator $L_p(\gamma_p)$ has a single negative eigenvalue. Finally, we return to the subcritical case and establish Lieb-Thirring-type bounds to eigenvalue moments.

%%%%%%%%%%%%%%%%%%%%%%%%%%%%%%%%%%%%%%%%%%%%%%%%
\section{Supercritical case} \label{s:supercrit}
\setcounter{equation}{0}

As indicated, our first main result is the following.

 % ------------- %
\begin{theorem} \label{thm: super}
For any $\lambda>\gamma_p$ we have $\sigma(L_p(\lambda))=\mathbb{R}$.
\end{theorem}
 % ------------- %
\begin{proof}
To demonstrate that any real number $\mu$ belongs to essential spectrum of operator $L_p$ we are going to use Weyl's criterion: we have to find a sequence $\{\psi_k\}_{k=1}^\infty\subset D(L_p)$ such that $\|\psi_k\|=1$ which contains no convergent subsequence and
 % ------------- %
$$
\|L_p\psi_k-\mu\psi_k\|\to 0\quad\text{as}\quad k\to\infty\,.
$$
 % ------------- %
For the sake of clarity let us first show that $0\in\sigma_{\mathrm{ess}}(L_p)$. We define
 % ------------- %
 \begin{equation}\label{sequence}
 \psi_k(x,y):=\frac{1}{k^{1/(p+2)}}\, h_p\left(x y^{p/(p+2)}\right)\,\e^{i\beta y^{(2p+2)/(p+2)}}\chi\left(\frac{y}{k}\right)\,,
 \end{equation}
 % ------------- %
where $h_p$ is the ground state eigenfunction of $H_p$, $\:\chi$ is a smooth function with $\mathrm{supp}\,\chi\subset[1, 2]$ satisfying $\int_1^2\chi^2(z)\,\mathrm{d}z=1$, and $\beta>0$ will be chosen later. We note that for a given $k$ one can achieve that $\|\psi_k\|_{L^2(\mathbb{R}^2)} \ge\frac{1}{2^{p/(p+2)}}$ as the following estimates show,
 % ------------- %
\begin{eqnarray}
\lefteqn{\int_{\mathbb{R}^2}\left|\frac{1}{k^{1/(p+2)}}\,h_p(x y^{p/(p+2)})\,\e^{i\beta y^{(2p+2)/(p+2)}}\chi\left(\frac{y}{k}\right)
\right|^2\,\mathrm{d}x\,\mathrm{d}y} \nonumber \\ && =\frac{1}{k^{2/(p+2)}}\int_k^{2k}\int_{\mathbb{R}}
\left|h_p(xy^{p/(p+2)})\,\chi\left(\frac{y}{k}\right)\right|^2\,\mathrm{d}x\,\mathrm{d}y \nonumber \\ && =\frac{1}{k^{2/(p+2)}}\int_k^{2k}\int_{\mathbb{R}}\frac{1}{y^{p/(p+2)}}\left|h_p(t)\,\chi\left(\frac{y}{k}\right)\right|^2\,\mathrm{d}t\,\mathrm{d}y
\nonumber \\ && =\frac{1}{k^{2/(p+2)}}\int_{\mathbb{R}}|h_p(t)|^2\,\mathrm{d}t\,\int_k^{2k}\frac{1}{y^{p/( p+2)}}\left|\chi\left(\frac{y}{k}\right)\right|^2\,\mathrm{d}y \nonumber \\ &&  =\frac{1}{k^{2/(p+2)}}\int_k^{2k}\frac{1}{y^{p/(p+2)}}\left|\chi\left(\frac{y}{k}\right)\right|^2\,\mathrm{d}y \nonumber \\ && \label{firstpart} \ge\frac{1}{2^{p/(p+2)}}\int_1^2|\chi(z)|^2\,\mathrm{d}z=\frac{1}{2^{p/(p+2)}}\,.
\end{eqnarray}
 % ------------- %
Our next aim is to show that for any positive $\varepsilon$ one can find $k=k(\varepsilon)$ such that  $\|L_p\psi_k\|_{L^2(\mathbb{R}^2)}^2 <\varepsilon$ holds. By a straightforward calculation one gets
 % ------------- %
$$
\frac{\partial^2\psi_k}{\partial x^2}=\frac{1}{k^{1/(p+2)}}\, y^{2p/(p+2)}\, h_p''(x y^{p/(p+2)})\,\e^{i\beta y^{(2p+2)/(p+2)}}\,\chi\left(\frac{y}{k}\right)
$$
 % ------------- %
and
 % ------------- %
\begin{eqnarray}
\lefteqn{\frac{\partial^2\psi_k}{\partial y^2}
=\frac{1}{k^{1/(p+2)}}\,\e^{i\beta y^{(2p+2)/(p+2)}}\biggl(-\frac{2px}{(p+2)^2}\, y^{-(p+4)/(p+2)}\,h_p'(x y^{p/(p+2)})\,\chi\left(\frac{y}{k}\right)}
\nonumber \\ && +\frac{p^2 x^2}{(p+2)^2}\, y^{-4/(p+2)}\,h_p''(x y^{p/(p+2)})\,\chi\left(\frac{y}{k}\right)
\nonumber \\ && +\frac{i p (4p+4) \beta x}{(p+2)^2}\,y^{(p-2)/(p+2)}\,h_p'(x y^{p/(p+2)})\,\chi\left(\frac{y}{k}\right)\nonumber \\ &&
+\frac{2px}{k( p+2)}y^{-2/(p+2)}\,h_p'(x y^{p/(p+2)})\,\chi'\left(\frac{y}{k}\right)
\nonumber \\ && +\frac{i \beta(2p+2)p}{(p+2)^2}\, y^{-2/(p+2)}\, h_p(x y^{p/(p+2)})\,\chi\left(\frac{y}{k}\right)
\nonumber \\ &&
+\frac{2i \beta(2p+2)}{(p+2)k}y^{p/(p+2)}\, h_p(x y^{p/(p+2)})\,\chi'\left(\frac{y}{k}\right)+\frac{1}{k^2}\,h_p(x y^{p/(p+2)})\,\chi''\left(\frac{y}{k}\right)\biggr)\nonumber \\ && -\frac{ \beta^2(2p+2)^2}{(p+2)^2}\,y^{2p/(p+2)}\,h_p(x y^{p/(p+2)})\,\chi\left(\frac{y}{k}\right). \label{calculations}
\end{eqnarray}
 % ------------- %
Our aim is to show that choosing $k$ sufficiently large one can make most terms at the right-hand side of (\ref{calculations}) as small as we wish. Changing the integration variables, we get for the first term the following estimate,
 % ------------- %
\begin{eqnarray*}
\lefteqn{\int_{\mathbb{R}^2}\left|\frac{x}{k^{1/(p+2)}\, y^{(p+4)/(p+2)}}\,h_p'(x y^{p/(p+2)})\,\e^{i\beta y^{(2p+2)/(p+2)}}\,\chi\left(\frac{y}{k}\right)\right|^2\,\mathrm{d}x\,\mathrm{d}y} \\ && =\frac{1}{k^{2/(p+2)}}\int_k^{2k}\int_{\mathbb{R}}
\left|\frac{x}{y^{(p+4)/(p+2)}}\,h_p'(x y^{p/(p+2)})\,\chi\left(\frac{y}{k}\right)\right|^2\,\mathrm{d}x\,\mathrm{d}y
 \\ && =\frac{1}{k^{2/(p+2)}}\int_k^{2k}\frac{1}{y^{(5p+8)/(p+2)}}
\left|\chi\left(\frac{y}{k}\right)\right|^2\,\mathrm{d}y\,\int_{\mathbb{R}}t^2\,|h_p'(t)|^2\,\mathrm{d}t \\ && \le\frac{1}{k^4}
\int_1^2|\chi(z)|^2\mathrm{d}z\,\int_{\mathbb{R}}t^2\,|h_p'(t)|^2\,\mathrm{d}t\,,
\end{eqnarray*}
 % ------------- %
where the right-hand side tends to zero as $k\to\infty$. In the same way we establish that for large enough $k$ all the terms in (\ref{calculations}) except the last one can be made small. The last term is not small, what is important that it asymptotically compensates with the negative part of the potential; using the same technique one can prove that for large $k$ the integral
 % ------------- %
$$
\frac{1}{k^{2/(p+2)}}\int_{\mathbb{R}^2}\biggl((x^2+y^2)^{p/(p+2)}-y^{2p/(p+2)}\biggr)^2 h_p^2(x y^{p/(p+2)})\,\chi^2\left(\frac{y}{k}\right)\,\mathrm{d}x\,\mathrm{d}y
$$
 % ------------- %
is small again as small as we wish. Consequently, for any fixed $\varepsilon>0$ one can choose $k$ large enough such that
 % ------------- %
\begin{eqnarray*}
\lefteqn{\int_{\mathbb{R}^2}|L_p\psi_k|^2(x,y)\,\mathrm{d}x\,\mathrm{d}y}  \\ &&
=\int_{\mathbb{R}^2}\left|-\frac{\partial^2\psi_k}{\partial x^2}-
\frac{\partial^2\psi_k}{\partial y^2}+|x y|^p\psi_k-\lambda (x^2+y^2)^{p/(p+2)}\psi_k\right|^2\,
\mathrm{d}x\,\mathrm{d}y \\ && \le\frac{1}{k^{2/(p+2)}}\int_{k}^{2k}\int_{\mathbb{R}}\biggl|y^{2p/(p+2)}\, h''_p(x y^{p/(p+2)})\chi\left(\frac{y}{k}\right) \\ && -\frac{\beta^2 (2p+2)^2}{(p+2)^2}\, y^{2p/(p+2)}\, h_p(x y^{p/(p+2)})\chi\left(\frac{y}
{k}\right) \\ &&  -|x y|^p\,h_p(x y^{p/(p+2)})\chi\left(\frac{y}{k}\right)
+\lambda y^{2p/(p+2)}\, h(x y^{p/(p+2)})\,\chi\left(\frac{y}{k}\right)\biggr|^2\,\mathrm{d}x\,\mathrm{d}y+\varepsilon \\ &&
=\frac{1}{k^{2/(p+2)}}\int_{k}^{2k}\int_{\mathbb{R}}\biggl|y^{2p/(p+2)}\biggl(h_p''(x y^{p/(p+2)})-|x y^{p/(p+2)}|^p\, h_p(x y^{p/(p+2)}) \\ && -\frac{\beta^2 (2p+2)^2}{(p+2)^2}\, h_p(x y^{p/(p+2)})+\lambda h_p(x y^{p/(p+2)})\biggr)\chi\left(\frac{y}{k}\right)\biggr|^2\,\mathrm{d}x\,\mathrm{d}y+\varepsilon\,.
\end{eqnarray*}
 % ------------- %
Combining this result with the fact that $H_p h_p= \gamma_p h_p$ and choosing
 % ------------- %
\begin{equation} \label{beta-super}
\beta=\frac{(p+2)}{2p+2}\sqrt{\lambda-\gamma_p}
\end{equation}
 % ------------- %
we get
 % ------------- %
\begin{equation}
\label{final}
\int_{\mathbb{R}^2}|L_p\psi_k|^2(x,y)\,\mathrm{d}x\,\mathrm{d}y\le\varepsilon\,.
\end{equation}
 % ------------- %
To complete this part of the proof we fix a sequence $\{\varepsilon_j\}_{j=1}^\infty$ such that $\varepsilon_j\searrow0$ holds as $j\to\infty$ and to any $j$ we construct a function $\psi_{k(\varepsilon_j)}$ such that the supports for different $j$'s do not intersect each other; this can be achieved by choosing each next $k(\varepsilon_j)$ large enough. The norms of $L_p\psi_{k(\varepsilon_j)}$ satisfy the inequality (\ref{final}) with $\varepsilon_j$ on the right-hand side, and by construction the sequence $\psi_{k(\varepsilon_j)}$ converges weakly to zero; this yields the sought Weyl sequence for zero energy.

Passing now to an arbitrary nonzero real number $\mu$ we can use the same procedure replacing the above functions $\psi_k$ by
 % ------------- %
\begin{equation} \label{superweyl}
\psi_k(x,y)=\frac{1}{k^{1/(p+2)}}\, h_p(x y^{p/(p+2)})\,\e^{i\epsilon_\mu(y)}\, \chi\left(\frac{y}{k}\right)\,,
\end{equation}
 % ------------- %
where
 % ------------- %
$$
\epsilon_\mu(y):= \displaystyle{\int_{\frac{|\mu|^{(p+2)/2p}(p+2)^{(p+2)/p}}{(2p+2)^{(p+2)/p}\beta^{(p+2)/p}}}^y\sqrt{\frac{(2p+2)^2 \beta^2}{(p+2)^2}\, t^{2p/(p+2)}+\mu}\:\mathrm{d}t}\,,
$$
 % ------------- %
and furthermore, the functions $h_p,\,\chi$ and the number $\beta$ are the same way as above. The second derivatives of those functions are
 % ------------- %
$$
\frac{\partial^2\psi_k}{\partial x^2}=\frac{1}{k^{1/(p+2)}}\, y^{2p/(p+2)}\, h_p''(x y^{p/(p+2)})\,\e^{i\epsilon_\mu(y)}\,
\chi\left(\frac{y}{k}\right)
$$
 % ------------- %
and
 % ------------- %
\begin{eqnarray*}
\lefteqn{\frac{\partial^2\psi_k}{\partial y^2}
=\frac{1}{k^{1/(p+2)}}\,\e^{i\epsilon_\mu(y)}\biggl(\frac{-2p x}{(p+2)^2}\, y^{-(p+4)/(p+2)}\,h_p'(x y^{p/(p+2)})\,\chi\left(\frac{y}{k}\right)} \\ && \hspace{-1.5em}  +\frac{p^2 x^2}{(p+2)^2}\,y^{-4/(p+2)}\,h_p''(x y^{p/(p+2)})\,\chi\left(\frac{y}{k}\right)+\frac{2p x}{k(p+2)}\,y^{-2/(p+2)}\,h_p'(x y^{p/(p+2)})\,\chi'\left(\frac{y}{k}\right) \\ && \hspace{-1.5em}
+i p\,\frac{(2p+2)^2}{(p+2)^3}\,\beta^2 y^{(p-2)/(p+2)}\left(\frac{(2p+2)^2\beta^2}{(p+2)^2}\,y^{2p/(p+2)}+\mu\right)^{-1/2} \,h_p(x y^{p/(p+2)})\,\chi\left(\frac{y}{k}\right) \\ && \hspace{-1.5em} +\frac{2i p x}{(p+2)}\,y^{-2/(p+2)}\left(\frac{(2p+2)^2\beta^2}{(p+2)^2}\,y^{2p/(p+2)}+\mu\right)^{1/2}h_p'(x y^{p/(p+2)})\,\chi\left(\frac{y}{k}\right) \\ && \hspace{-1.5em}
+\frac{2i}{k}\left(\frac{(2p+2)^2\beta^2}{(p+2)^2}\,y^{2p/(p+2)}+\mu\right)^{1/2}\, h_p(x y^{p/(p+2)})\,\chi'\left(\frac{y}{k}\right) \\ && \hspace{-1.5em} +\frac{1}{k^2}\,h_p(x y^{p/(p+2)})\,\chi''\left(\frac{y}{k}\right)
-\left(\frac{(2p+2)^2\beta^2}{(p+2)^2}\,y^{2p/(p+2)}+\mu\right) h_p(x y^{p/(p+2)})\,\chi\left(\frac{y}{k}\right)\biggr).
\end{eqnarray*}
 % ------------- %
It is not difficult to check that for any positive $\varepsilon$ one choose a number $k$ large enough to ensure that the inequality
 % ------------- %
\begin{eqnarray*}
\lefteqn{\biggl\|\frac{\partial^2\psi_k}{\partial y^2} \,\e^{-i\epsilon_\mu(y)}
+\mu\psi_k \,\e^{-i\epsilon_\mu(y)}} \\ && - \e^{-i\beta y^{(2p+2)/(p+2)}}\frac{\partial^2}{\partial y^2}\biggl(\psi_k\, \e^{-i\epsilon_\mu(y)+i\beta y^{(2p+2)/(p+2)}}\biggr)\biggr\|_{L^2(\mathbb{R}^2)}
<\varepsilon
\end{eqnarray*}
% ------------- %
holds. Using further the identity
 % ------------- %
$$
\frac{\partial^2\psi_k}{\partial x^2}\,\e^{-i\epsilon_\mu(y)}= e^{-i\beta y^{(2p+2)/(p+2)}}
\frac{\partial^2}{\partial x^2}\bigl(\psi_k\,\e^{-i\epsilon_\mu(y)+i\beta y^{(2p+2)/(p+2)}}\bigr)
$$
 % ------------- %
we arrive at the estimate
 % ------------- %
\begin{eqnarray*}
\lefteqn{\|L_p\psi_k-\mu\psi_k\|_{L^2(\mathbb{R}^2)}=\bigg\|(L_p\psi_k) e^{-i\epsilon_\mu(y)}
-\mu\psi_k \,e^{-i\epsilon_\mu(y)}\biggr\|_{L^2(\mathbb{R}^2)} } \\ &&  <\biggl\| \e^{i\beta y^{(2p+2)/(p+2)}}L_p\biggl(\psi_k \,\e^{-i\epsilon_\mu(y)+i\beta y^{(2p+2)/(p+2)}}\biggr)\biggr\|_{L^2(\mathbb{R}^2)}+\varepsilon\,;
\end{eqnarray*}
 % ------------- %
now we can use the result of the first part of proof to establish the claim.
\end{proof}

%%%%%%%%%%%%%%%%%%%%%%%%%%%%%%%%%%%%%%%%%%%
\section{Critical case}  \label{s:critical}
\setcounter{equation}{0}

Let us now pass to the case when the parameter value is critical, in other words, consider the operator $L_p(\gamma_p)=-\Delta +(|x y|^p-\gamma_p(x^2+y^2)^{p/(p+2)}),\:p\ge1$, on $L^2(\mathbb{R}^2)$. We shall consider the positive and negative spectrum separately.

\subsection{The essential spectrum}
\setcounter{equation}{0}

First we are going to show that the discreteness is lost in the positive halfline once the coupling constant reaches the critical value.

 % ------------- %
\begin{theorem} \label{thm: crit-ess}
The essential spectrum of $L_p(\gamma_p)$ contains the interval $[0, \infty)$.
\end{theorem}
 % ------------- %
\begin{proof}
The argument is similar to that used in the proof of Theorem\ref{thm: super}, hence we present it briefly with emphasis on the differences. As before we check first that $0\in\sigma_{\mathrm{ess}}(L_p)$ by constructing a Weyl sequence, which is now of the form
 % ------------- %
 $$
 \psi_k(x,y):=\frac{1}{k^{1/(p+2)}}\, h_p\left(x y^{p/(p+2)}\right)\,\chi\left(\frac{y}{k}\right)
 $$
 % ------------- %
with $h_p$ and $\chi$ the same as before. As this nothing but \eqref{sequence} with $\beta=0$, not surprisingly in view of \eqref{beta-super} we can repeat the reasoning with the involved expressions appropriately simplified.

Passing now to an arbitrary nonnegative number $\mu$ we replace \eqref{superweyl} by
 % ------------- %
$$
\psi_k(x,y)=\frac{1}{k^{1/(p+2)}}\, h_p(x y^{p/(p+2)})\,\e^{i \eta_\mu (y)}\, \chi\left(\frac{y}{k}\right)\,,
$$
 % ------------- %
where the functions $h_p,\,\chi$ are again the same way as above and $(\eta^\prime_\mu(y))^2=\mu$. This can be achieved for any $\mu\ge0$ t by choosing $\eta_\mu(y)=\sqrt{\mu} y$, note that the classically allowed region is now the whole halfline instead of the interval entering the definition of $\epsilon_\mu(y)$ above. The second derivatives of the functions $\psi_k$ obtained in this way are
 % ------------- %
$$
\frac{\partial^2\psi_k}{\partial x^2}=\frac{1}{k^{1/(p+2)}}\, y^{2p/(p+2)}\, h_p''(x y^{p/(p+2)})\,\e^{i \sqrt{\mu} y}\,
\chi\left(\frac{y}{k}\right)
$$
 % ------------- %
and
 % ------------- %
\begin{eqnarray*}
\lefteqn{\frac{\partial^2\psi_k}{\partial y^2}
=\frac{1}{k^{1/(p+2)}}\,\e^{i \sqrt{\mu} y}\biggl(\frac{-2p x}{(p+2)^2}\, y^{-(p+4)/(p+2)}\,h_p'(x y^{p/(p+2)})\,\chi\left(\frac{y}{k}\right)} \\ && \hspace{-1.5em}  +\frac{p^2 x^2}{(p+2)^2}\,y^{-4/(p+2)}\,h_p''(x y^{p/(p+2)})\,\chi\left(\frac{y}{k}\right)+\frac{2p x}{k(p+2)}\,y^{-2/(p+2)}\,h_p'(x y^{p/(p+2)})\,\chi'\left(\frac{y}{k}\right) \\ && \hspace{-1.5em}
+\frac{2i \sqrt{\mu} p x}{(p+2)}\,y^{-2/(p+2)}h_p'(x y^{p/(p+2)})\,\chi\left(\frac{y}{k}\right) \\ && \hspace{-1.5em}
+\frac{2i \sqrt{\mu}}{k}\, h_p(x y^{p/(p+2)})\,\chi'\left(\frac{y}{k}\right) \\ && \hspace{-1.5em} +\frac{1}{k^2}\,h_p(x y^{p/(p+2)})\,\chi''\left(\frac{y}{k}\right)
-\mu h_p(x y^{p/(p+2)})\,\chi\left(\frac{y}{k}\right)\biggr).
\end{eqnarray*}
 % ------------- %
One finds easily that for any positive $\varepsilon$ and $k$ large enough we have
 % ------------- %
$$
\biggl\|\frac{\partial^2\psi_k}{\partial y^2} \,\e^{-i \sqrt{\mu} y}
+\mu\psi_k \,\e^{-i \sqrt{\mu} y} - \frac{\partial^2}{\partial y^2}\biggl(\psi_k\, \e^{-i \sqrt{\mu} y}\biggr)\biggr\|_{L^2(\mathbb{R}^2)}<\varepsilon
$$
% ------------- %
and using further the trivial identity $\frac{\partial^2\psi_k}{\partial x^2}\,\e^{-i \sqrt{\mu} y}= \frac{\partial^2}{\partial x^2}\bigl(\psi_k\,\e^{-i \sqrt{\mu} y}\bigr)$ we arrive at
 % ------------- %
$$
\|L_p\psi_k-\mu\psi_k\|_{L^2(\mathbb{R}^2)}=\Big\|\big(L_p\psi_k
-\mu\psi_k\big) \,e^{-i \sqrt{\mu} y}\Big\|_{L^2(\mathbb{R}^2)} <\Big\| L_p\Big(\psi_k \,\e^{-i \sqrt{\mu} y}\bigr)\Big\|_{L^2(\mathbb{R}^2)}+\varepsilon
$$
 % ------------- %
and the result of the first part of proof allows us to establish the claim.
\end{proof}

\subsection{Discreteness of the negative spectrum}
\setcounter{figure}{0}

Next we are going to show that the inclusion $\sigma_\mathrm{ess}(L_p(\gamma_p)) \supset [0, \infty)$ established in Theorem~\ref{thm: crit-ess} is in fact an equality.

 % ------------- %
\begin{theorem} \label{thm:subcrit}
The negative spectrum of $L_p(\gamma_p),\:p\ge1$, is discrete.
\end{theorem}
 % ------------- %
\begin{proof}[Proof]
By the minimax principle it is sufficient to estimate $L_p$ from below by a self-adjoint operator with a purely discrete negative spectrum. To construct such a lower bound we employ a bracketing argument, imposing additional Neumann conditions at the rectangles $G_n=\{-\alpha_{n+1}<x<\alpha_{n+1}\}\times \{\alpha_n< y<\alpha_{n+1}\}$, $\:\widetilde{G}_n=\{-\alpha_{n+1}<x<\alpha_{n+1}\}\times \{-\alpha_{n+1}< y<-\alpha_n\}$, $\:Q_n=\{\alpha_n< x<\alpha_{n+1}\}\times \{-\alpha_n<y<\alpha_n\}$, and $\widetilde{Q}_n=\{-\alpha_{n+1}< x<-\alpha_n\}\times \{-\alpha_n<y<\alpha_n\},\: n=1,2,\ldots$, together with central square $G_0=(-\alpha_1, \alpha_1)^2$ -- cf.~Fig.~1. Here $\{\alpha_n\}_{n=1}^\infty$ is a monotone sequence such that $\alpha_n\to\infty$ as $n\to\infty$ which will be specified later. In this way we obtain a direct sum of operators with Neumann boundary conditions at the rectangle boundaries which we denote as
 % ------------- %
 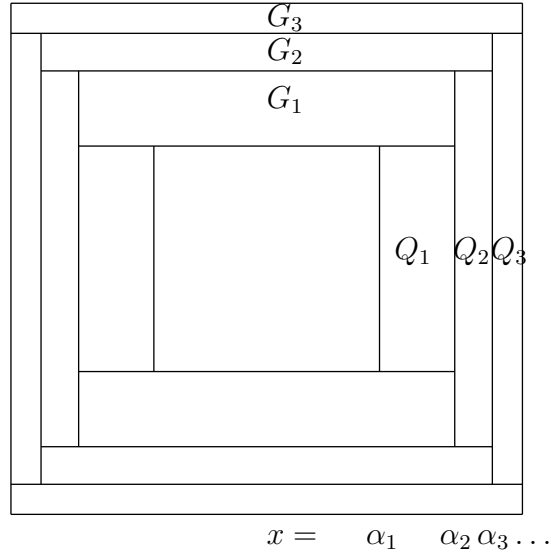
\begin{figure}
 \begin{center}
 \setlength\unitlength{1mm}
 \begin{picture}(120,80)(27,0)
 %\thinlines
 \linethickness{.5pt}
 \put(56,70){\line(1,0){68}}
 \put(60,65){\line(1,0){60}}
 \put(60,15){\line(1,0){60}}
 \put(56,10){\line(1,0){68}}
 \put(60,10){\line(0,1){60}}
 \put(120,10){\line(0,1){60}}
 \put(65,15){\line(0,1){50}}
 \put(115,15){\line(0,1){50}}
 \put(65,55){\line(1,0){50}}
 \put(65,25){\line(1,0){50}}
 \put(75,25){\line(0,1){30}}
 \put(105,25){\line(0,1){30}}
 \put(56,74){\line(1,0){68}}
 \put(56,6){\line(1,0){68}}
 \put(56,6){\line(0,1){68}}
 \put(124,6){\line(0,10){68}}
 \put(90,60){$G_1$}
 \put(90,66.5){$G_2$}
 \put(90,71){$G_3$}
 \put(107,40){$Q_1$}
 \put(114.8,40){$Q_2$}
 \put(119.8,40){$Q_3$}
 \put(90,2){$x=\quad\;\alpha_1$}
 \put(113,2){$\alpha_2$}
 \put(118,2){$\alpha_3$}
 \put(123,2){$\dots$}
 \end{picture}
 \caption{The Neumann bracketing scheme}
 \end{center}
 \end{figure}
 % ------------- %

 % ------------- %
$$
L^{(1)}_{n, p}=L_p|_{G_n}, \quad \widetilde{L}^{(1)}_{n, p}=L_p|_{\widetilde{G}_n},\quad L^{(2)}_{n, p}=L_p|_{Q_n}, \quad \widetilde{L}^{(2)}_{n, p}=L_p|_{\widetilde{Q}_n}
$$
 % ------------- %
and $L^0_p=L_p|_{G_0}$. It is obvious that the spectra of $L_{n, p}^{(i)},\;\tilde{L}^{(i)}_{n, p},\; i=1, 2$, and $L_p^0$ are purely discrete, hence one needs to check that $\underline{\lim}_{n\to\infty}\,\inf\,\sigma\big(L^{(i)}_{n,p}\big)\ge 0$ and $\underline{\lim}_{n\to\infty}\,\inf\, \big(\sigma(\tilde{L}^{(i)}_{n, p}\big)\ge 0$ holds for $i=1,2$, since then the spectra of all the direct sums $\bigoplus_{n=1}^\infty L^{(i)}_{n, p}$ and $\bigoplus_{n=1}^\infty \tilde{L}^{(i)}_{n, p},\; i=1,2$, below any fixed negative number contain a finite number of eigenvalues, the multiplicity taken into account, which implies the sought claim.

Furthermore, the goal will be achieved if we estimate $L^{(i)}_{n, p},\; \tilde{L}_{n, p}^{(i)},\; i=1,2$, from below by operators with separated variables and prove the analogous limiting relations for them. We use the lower bounds
 % ------------- %
\begin{equation}\label{Hnp1}
H^{(1)}_n\psi =-\Delta \psi+ (\alpha_n^p |x|^p-\gamma_p(x^2+\alpha_{n+1}^2)^{p/(p+2)})\psi
\end{equation}
 % ------------- %
on $L^2(-\alpha_{n+1}, \alpha_{n+1}) \otimes L^2(\alpha_n, \alpha_{n+1})$ with the boundary conditions
 % ------------- %
$$
\left. \frac{\partial\psi}{\partial x}\right|_{x=-\alpha_{n+1}}= \left.
\frac{\partial\psi}{\partial x}\right|_{x=\alpha_{n+1}}=0 \,,
$$$$
\left. \frac{\partial\psi}{\partial y}\right|_{y=\alpha_n}= \left.
\frac{\partial\psi}{\partial y}\right|_{y=\alpha_{n+1}}=0 \,.
$$
 % ------------- %
It is clear that the spectra of $H^{(1)}_{n, p},\;n=1,2,\ldots$, are purely discrete; we are going to check that
 % ------------- %
\begin{equation}\label{lim}
\underline{\lim}_{n\to\infty}\inf\,\sigma\big(H_{n, p}^{(1)}\big)\ge 0\,.
\end{equation}
 % ------------- %
Since the lowest Neumann eigenvalue of $-\frac{\mathrm{d}^2}{\mathrm{d}y^2}$ on the interval is zero corresponding to a constant eigenfunction, the problem reduces to analysis of the  operator $h^{(1)}_{n, p}=-\frac{\mathrm{d}^2}{\mathrm{d}x^2}+\alpha_n^p |x|^p-\gamma_p (x^2+ \alpha_{n+1}^2)^{p/(p+2)}$ on $L^2(-\alpha_{n+1}, \alpha_{n+1})$. Using a simple scaling transformation, one can check that $h_{n, p}^{(1)}$ is unitarily equivalent to
 % ------------- %
\begin{equation}\label{un.eq.}
h_{n,p}^{(2)}=\alpha_n^{2p/(p+2)}\,\left(-\frac{\mathrm{d}^2}{\mathrm{d}x^2}+|x|^p
-\frac{\gamma_p}{\alpha_n^{2p/(p+2)}}\left(\frac{x^2}{\alpha_n^{2p/(p+2)}}
+\alpha_{n+1}^2\right)^{p/(p+2)}\right)
\end{equation}
 % ------------- %
on the interval $\big(-\alpha_{n+1}\,\alpha_n^{p/(p+2)}, \alpha_{n+1}\,\alpha_n^{p/(p+2)}\big]$ with Neumann boundary conditions at its endpoints. To proceed we need to specify the sequence $\{\alpha_n\}$. Let us assume that
 % ------------- %
\begin{equation}\label{Decay}
\alpha_{n+1}^{2p/(p+2)}-\alpha_n^{2p/(p+2)}\to 0\quad\text{as}\quad n\to\infty\,.
\end{equation}
 % ------------- %
Combining this assumption with the inequality
 % -------------%
 \begin{eqnarray*}
\lefteqn{\frac{\gamma_p}{\alpha_n^{2p/(p+2)}}\left(\left(\frac{x^2}{\alpha_n^{2p/(p+2)}}+\alpha_{n+1}^2\right)^{p/(p+2)}
-\alpha_{n+1}^{2p/(p+2)}\right)} \\ && \le\frac{\gamma_p}{\alpha_n^{2p/(p+2)}}\left(\frac{x^2}{\alpha_n^{2p/(p+2)}}\right)^{p/(p+2)} \le\frac{\gamma_p}{\alpha_n^{4p(p+1)/(p+2)^2}}\: (|x|^p+1)\,,
\end{eqnarray*}
 % ------------- %
we infer that
 % ------------- %
\begin{eqnarray}
\lefteqn{h_{n,p}^{(2)}=\alpha_n^{2p/(p+2)}\,\biggl(-\frac{\mathrm{d}^2}{\mathrm{d}x^2}
+|x|^p-\frac{\gamma_p\,\alpha_{n+1}^{2p/(p+2)}}{\alpha_n^{2p/(p+2)}}} \nonumber \\ && \quad
-\frac{\gamma_p}{\alpha_n^{2p/(p+2)}}\biggl(\left(\frac{x^2}{\alpha_n^{2p/(p+2)}}+\alpha_{n+1}^2\right)^{p/(p+2)} -\alpha_{n+1}^{2p/(p+2)}\biggr)\biggr) \nonumber \\ && \ge\alpha_n^{2p/(p+2)} \Bigg(-\frac{\mathrm{d}^2}{\mathrm{d}x^2}+\left(1-\frac{\gamma_p}{\alpha_n^{4p(p+1)/(p+2)^2}}\right)|x|^p \nonumber \\ && \quad -\frac{\gamma_p}{\alpha_n^{4p(p+1)/(p+2)^2}}
-\frac{\gamma_p \alpha_{n+1}^{2p/(p+2)}}{\alpha_n^{2p/(p+2)}}\Bigg) \nonumber \\ && \ge
\alpha_n^{2p/(p+2)} \Bigg(-\frac{\mathrm{d}^2}{\mathrm{d}x^2}
+\Bigg(1-\frac{\gamma_p}{\alpha_n^{4p(p+1)/(p+2)^2}}\Bigg)|x|^p-\gamma_p \nonumber \\ && \quad -\gamma_p\Bigg(\frac{\alpha_{n+1}^{2p/(p+2)}
-\alpha_n^{2p/(p+2)}}{\alpha_n^{2p/(p+2)}}\Bigg)\Bigg) +o(1) \nonumber \\ && \ge\alpha_n^{2p/(p+2)}\left(-\frac{\mathrm{d}^2}{\mathrm{d}x^2}+\left(1-\frac{\gamma_p}{\alpha_n^{4p(p+1)/(p+2)^2}}\right)|x|^p-\gamma_p\right)+o(1)\nonumber \\ && \label{As.}\ge\alpha_n^{2p/(p+2)}\left(1-\frac{\gamma_p}{\alpha_n^{4p(p+1)/(p+2)^2}}\right) \left(-\frac{\mathrm{d}^2}{\mathrm{d}x^2}+|x|^p-\gamma_p\right)+o(1)\,,
\end{eqnarray}
 % ------------- %
where  the corresponding Neumann (an)harmonic oscillator is restricted to the interval
 % ------------- %
$$
(-\alpha_{n+1}\,\alpha_n^{p/(p+2)}, \alpha_{n+1}\,\alpha_n^{p/(p+2)})\,.
$$
 % ------------- %
Next we need to establish the following lemma.
 % ------------- %
\begin{lemma} \label{lemma ground}
Let  $l_{k, p}=-\frac{\mathrm{d}^2}{\mathrm{d}x^2}+|x|^p$ be the Neumann operator defined on the interval  $[-k, k],\;k>0$. Then
 % ------------- %
\begin{equation}\label{hnp}
\inf\,\sigma\left(l_{k, p}\right)\ge\gamma_p+o\left(\frac{1}{k^{p/2}}\right) \;\quad\text{as}\quad k\to\infty\,.
\end{equation}
 % ------------- %
\end{lemma}
 % ------------- %
\begin{proof}
The relation (\ref{hnp}) is certainly valid if $\inf \sigma\left(l_{k, p}\right)\ge\gamma_p$ holds for all $k$ from some number on. Assume thus that
we have $\inf \sigma\left(l_{k, p}\right)<\gamma_p$ for infinitely many numbers $k$. Let $\psi_{k, p}$ be the normalized ground-state eigenfunction of $l_{k, p}$. We fix a positive $\delta$ and check that
 % ------------- %
\begin{eqnarray}
\int_{-k}^{-k+1}\left(|\psi_{k, p}'|^2+|x|^p|\psi_{k, p}|^2\right)\,\mathrm{d}x<\delta\,, \nonumber
\\[-.7em]\label{eqn}\\[-.7em]
\int_{k-1}^k\left(|\psi_{k, p}'|^2+|x|^p|\psi_{k, p}|^2\right)\,\mathrm{d}x<\delta\,. \nonumber
\end{eqnarray}
 % ------------- %
Indeed, suppose that at least one of inequalities (\ref{eqn}) does not hold, then
 % ------------- %
\begin{equation}\label{middle}
\int_{-k+1}^{k-1}\left(|\psi_{k, p}'|^2+|x|^p|\psi_{k, p}|^2\right)\,\mathrm{d}x<\gamma_p-\delta\,.
\end{equation}
 % ------------- %
Since $\psi_{k, p}$ is by assumption the ground-state eigenfunction of $l_{k, p}$, we have
 % ------------- %
$$
\inf \sigma (l_{k, p})=\int_{-k}^k\left(|\psi_{k, p}'|^2+|x|^p|\psi_{k, p}|^2\right)\,\mathrm{d}x \le\int_{-1}^1\left(|\phi'|^2+|x|^p|\phi|^2\right)\,\mathrm{d}x
$$
 % ------------- %
for all $k\ge 1$ and any normalized function $\phi$ from the domain of the operator, in particular, for any $\phi$ from the class $C_0^\infty(-1,1)$ such that $\int_{-1}^1|\phi|^2\,\mathrm{d}x=1$. Consequently,  for large enough $k$ there must exist points $x_{k, p}^{(1)}\in(-k+1, -k+2)$ and $x_{k, p}^{(2)}\in(k-2, k-1)$ such that
 % ------------- %
$$
\psi_{k, p}\left(x_{k, p}^{(1)}\right)=\mathcal{O}\Big(\frac{1}{k^{p/2}}\Big)\quad\text{and}\quad  \psi_{k, p}\left(x_{k, p}^{(2)}\right)=\mathcal{O}\Big(\frac{1}{k^{p/2}}\Big)\;\quad\text{as}\quad k\to\infty\,.
$$
 % ------------- %
Next we construct a function $\varphi_{k, p}$ on semi-infinite intervals $(-\infty, x_{k, p}^{(1)})$ and $(x_{k, p}^{(2)},
\infty)$ in such a way that
 % ------------- %
$$
g_{k, p}(x):=\psi_{k, p}(x)\chi_{(x_{k, p}^{(1)}, x_{k, p}^{(2)})}(x)+\varphi_{k, p}(x)\chi_{(-\infty, x_{k, p}^{(1)})\cup(x_{k, p}^{(2)}, \infty)}(x)\in\mathcal{H}^{1}(\mathbb{R})
$$
 % ------------- %
and
 % ------------- %
\begin{equation}\label{condition}
\int_{-\infty}^{x_{k, p}^{(1)}}\left(|\varphi_{k, p}'|^2+|x|^p|\varphi_{k, p}|^2\right)\,\mathrm{d}x+\int_{x_{k, p}^{(2)}}^\infty\left(|\varphi_{k, p}'|^2+|x|^p|\varphi_{k, p}|^2\right)\,\mathrm{d}x =\mathcal{O}\Big(\frac{1}{k^{p/2}}\Big)\,;
\end{equation}
 % ------------- %
this can be always achieved, one can take, e.g., the function decreasing linearly with respect to $|x-x_{k, p}^{(j)}|$ from the values $\psi_{k, p}\left(x_{k, p}^{(j)}\right),\; j=1,2,$ to zero. By virtue of (\ref{middle}) and (\ref{condition}) we then have
 % ------------- %
$$
\int_{\mathbb{R}}\left(|g_{k, p}|^2+|x|^p|g_{k, p}|^2\right)\,\mathrm{d}x<\gamma_p-\delta+\mathcal{O}\left(\frac{1}{k^{p/2}}\right)<\gamma_p
$$
 % ------------- %
for large enough $k$, however, this is in contradiction with the fact that $\gamma_p$ is the ground-state eigenvalue of $l_{k,p}$. This  proves the validity of (\ref{eqn}).

Having established the validity of inequalities (\ref{eqn}) we infer from them that there are points $y^{(1)}_{k, p}\in(-k, -k+1)$ and $y_{k, p}^{(2)}\in(k-1, k)$ such that
 % ------------- %
$$
\psi_{k, p}(y^{(j)}_{k, p})=\mathcal{O}\Big(\frac{\delta}{k^{p/2}}\Big)\,, \quad j=1,2\,.
$$
 % ------------- %
Now we repeat the argument and construct a function $\tilde{\varphi}_{k, p}$ on the semi-infinite intervals $(-\infty, y^{(1)}_{k, p})$ and $(y^{(2)}_{k, p}, \infty)$ in such a way that
 % ------------- %
$$\tilde{g}_{k, p}(x):=\psi_{k, p}(x)\chi_{(y_{k, p}^{(1)}, y_{k, p}^{(2)})}(x)+\tilde{\varphi}(x)\chi_{(-\infty, y_{k, p}^{(1)})\cup(y_{k, p}^{(2)}, \infty)}(x)\in\mathcal{H}^{1}(\mathbb{R})
$$
 % ------------- %
and
 % ------------- %
$$
\int_{-\infty}^{y_{k, p}^{(1)}}\left(|\tilde{\varphi}_{k, p}'|^2+|x|^p|\tilde{\varphi}_{k, p}|^2\right)\,\mathrm{d}x+\int_{y_{k, p}^{(2)}}^\infty\left(|\tilde{\varphi}_{k, p}'|^2+|x|^p|\tilde{\varphi}_{k, p}|^2\right)\,\mathrm{d}x=\mathcal{O}\Big(\frac{\delta}{k^{p/2}}\Big)\,.
$$
 % ------------- %
Using the last relation one finds that
 % ------------- %
$$
\int_{\mathbb{R}}|\tilde{g}_{k, p}'|^2\,\mathrm{d}x+\int_{\mathbb{R}}|x|^p|\tilde{g}_{k, p}|^2\,\mathrm{d}x<\inf\,\sigma\left(l_{k, p}\right)+ \mathcal{O}\Big(\frac{\delta}{k^{p/2}}\Big)\,.
$$
 % ------------- %
However, $\gamma_p$ is the ground-state eigenvalue,
 % ------------- %
$$
\int_{\mathbb{R}}|\tilde{g}_{k, p}'|^2\,\mathrm{d}x+\int_{\mathbb{R}}|x|^p|\tilde{g}_{k, p}|^2\,\mathrm{d}x\ge\gamma_p\,,
$$
 % ------------- %
which in combination with above inequality gives
 % ------------- %
$$
\inf\,\sigma\left(l_{k, p}\right)>\gamma_p- \mathcal{O}\Big(\frac{\delta}{k^{p/2}}\Big)\,,
$$
 % ------------- %
proving the claim of the lemma.
\end{proof}

It follows from Lemma~\ref{lemma ground} that the right-hand side of the estimate (\ref{As.}) behaves asymptotically as
 % ------------- %
$$
o\left(\frac{1}{\alpha_n^{p(p+1)/(p+2)}}\right) \alpha_n^{2p/(p+2)}+o(1)
$$
 % ------------- %
which can be made arbitrarily small by choosing $n$ is large enough; this is what we needed to conclude the proof of Theorem~\ref{thm:subcrit}.
\end{proof}

\begin{remark}
{\rm We know from \cite[Thm.~2.1]{EB12} that the critical operator $L_p(\gamma_p)$ is bounded from below. Estimating separately the contributions to the respective quadratic form coming from the regions $\{(x,y):\: |y|\ge 1\}$, $\{(x,y):\: |x|\ge 1\,, |y|\le 1\}$, and the central square $(-1,1)^2$, we can derive a lower bound to the threshold of the negative spectrum in terms of spectral properties of the one-dimensional operators with the symbol
 % ------------- %
$$
-\frac{\mathrm{d}^2}{\mathrm{d}t^2}+|t|^p-\gamma_p\left(\frac{t^2}{z^{(4p+4)/(p+2)}}+1\right)^{p/(p+2)}
$$
 % ------------- %
with $z\ge1$. As such a bound is not simple and does not provide any significant insight, however, we are not going to present it here.}
\end{remark}

\subsection{Existence of the negative spectrum: a numerical indication}

Theorem~\ref{thm:subcrit} tells us that the spectrum in the negative halfline can be discrete only, and as we have remarked above one can find a lower estimate to its threshold, however, neither of these results implies anything about the negative spectrum \emph{existence}. Now we are going address this question numerically and provide an evidence of the discrete spectrum nontriviality.

We considering first the operator $L_2(\gamma_2)$ --- recall that $\gamma_2=1$ --- and impose a cutoff at a circle of radius $R$ circled at the origin with Dirichlet and Neumann boundary condition, and find the corresponding first and second eigenvalue using the Finite Element Method. The result is shown on Fig.~\ref{cutoff}.
 % ------------- %
\begin{figure}[h]\large
\begin{center}
\scalebox{0.8}{ \includegraphics{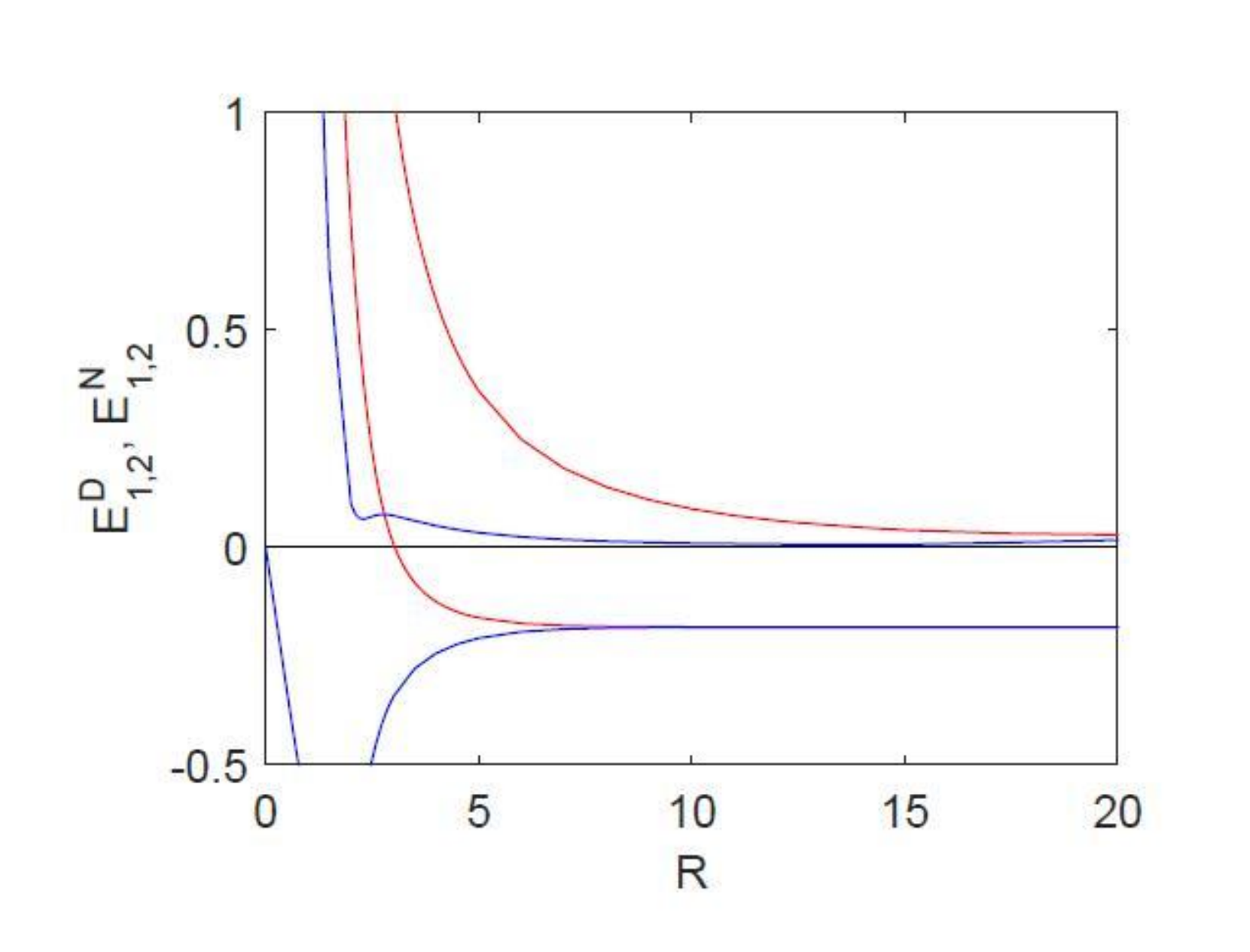}}
\caption{The eigenvalues $E_j,\: j=1,2$, of the critical operator with p=2 as functions of the cutoff radius $R$. The blue and red curves correspond to the Neumann and Dirichlet boundary, respectively.}\label{cutoff}
\end{center}
\end{figure}
 % ------------- %
We see, in particular, that the lowest Dirichlet eigenvalue is for $R\gtrsim 7$ practically independent of the cutoff radius and negative which by an elementary bracketing argument indicates that $L_2(1)$ has a negative eigenvalue. Furthermore, the difference between the Dirichlet and Neumann eigenvalue becomes negligible for large enough $R$ which shows that true ground-state eigenvalue in this case is $E\approx -0.18365$. For the second eigenvalue the DN gap also squeezes, although much slower and the Neumann eigenvalue is positive which hints that the discrete spectrum consists of a single point.
 % ------------- %
\begin{figure}[h]\large
\begin{center}
\scalebox{0.7}{\includegraphics{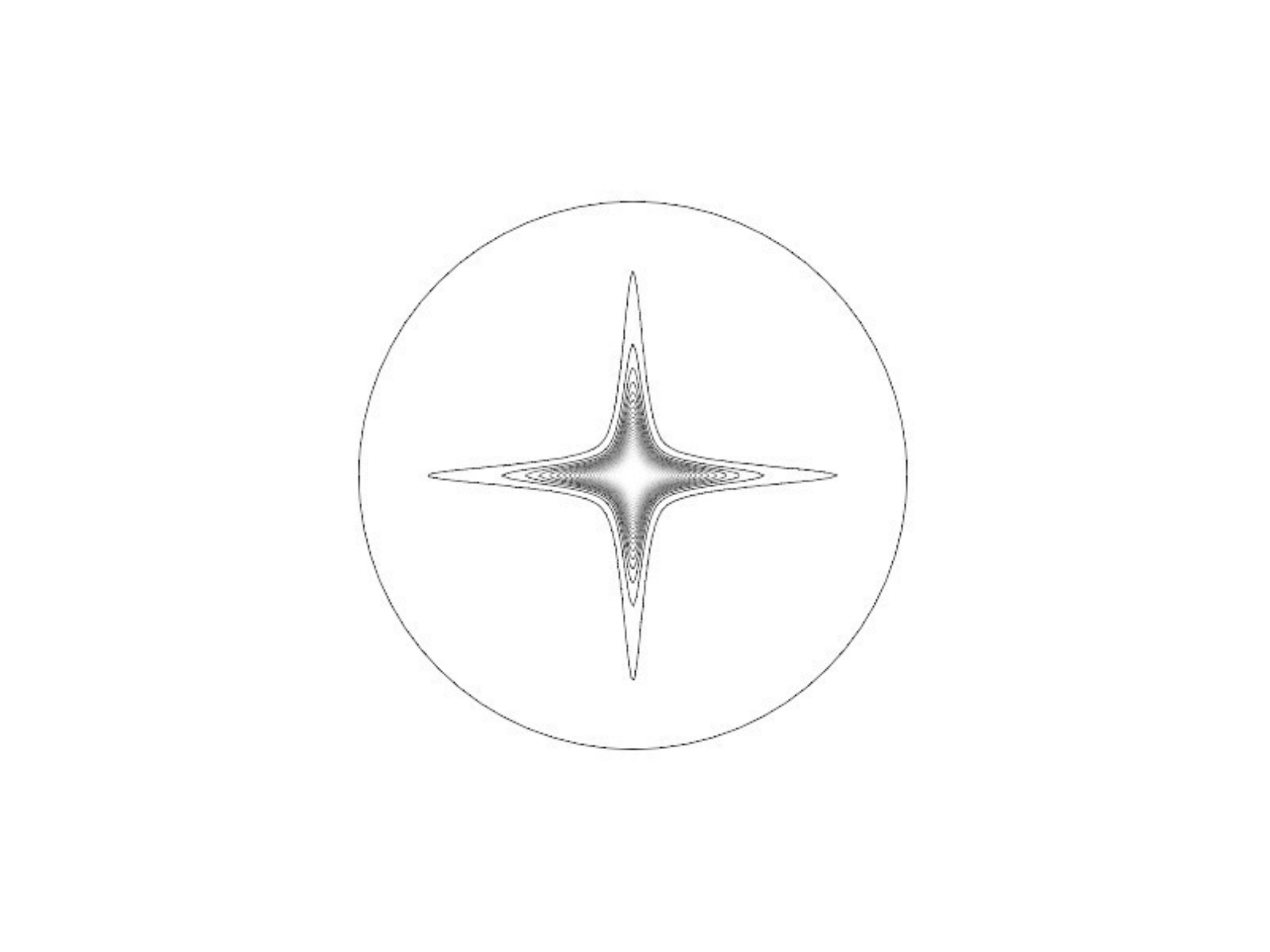}}
\caption{The ground-state eigenfunction for $p=2$, view from the top.}\label{groundstate}
\end{center}
\end{figure}
 % ------------- %
The Finite Element Method allows us also to compute the ground-state eigenfunction as shown on Fig.~\ref{groundstate}. The result is practically independent of the boundary condition used which is understandable since the function has an exponential falloff and the influence of the boundary is negligible for large enough $R$.

By continuity, the ground-state eigenvalue of $L_p(\lambda)$ exists in the vicinity of the point $p=2$; one is naturally interested what one can say about a broader range of the parameter. 
 % ------------- %
\begin{figure}[h]\large
\begin{center}
\scalebox{0.55}{\includegraphics{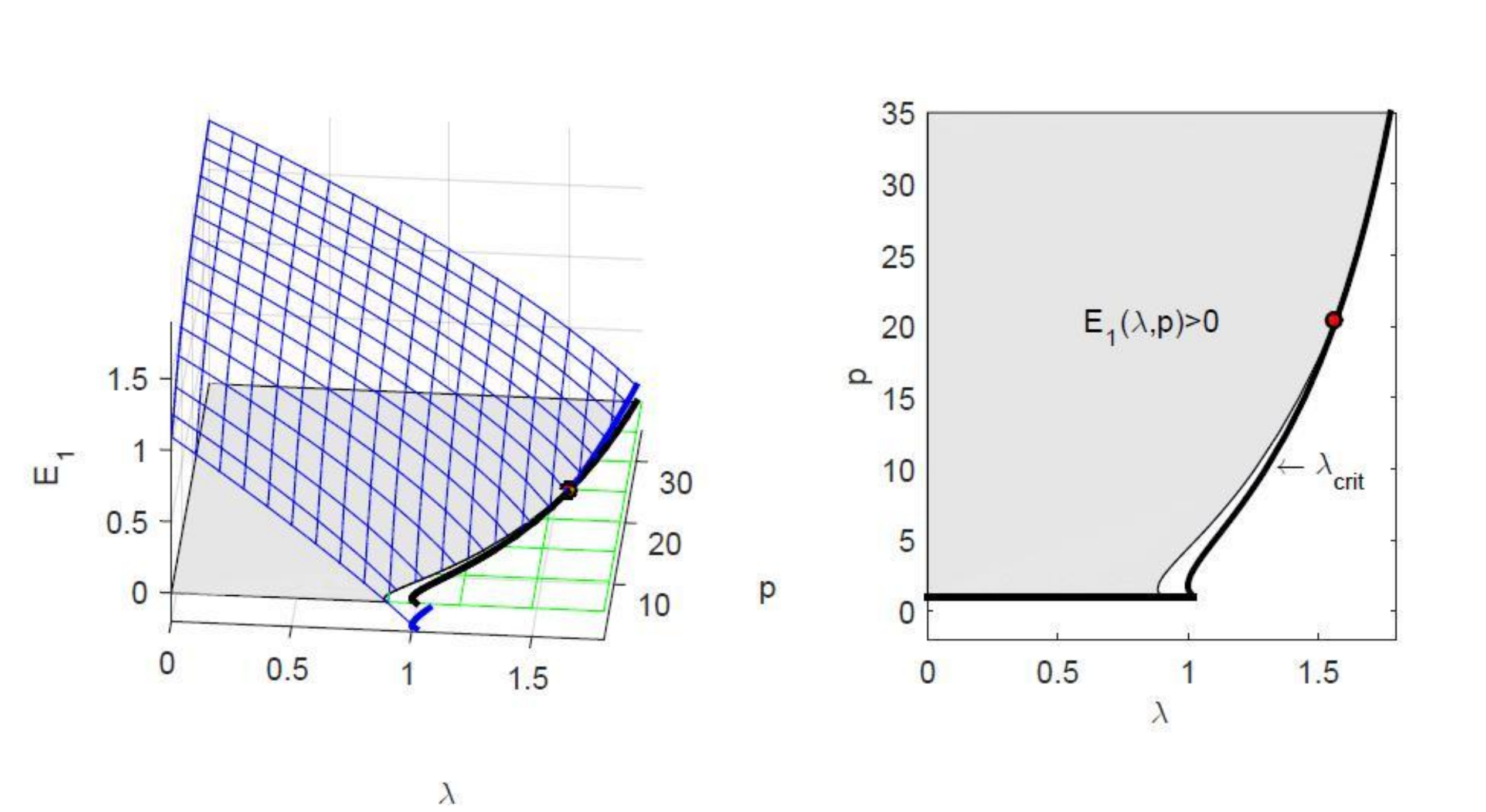}}
\caption{Positivity of $L_p(\lambda)$ as a function of $\lambda$ and $p$.}\label{critical3d}
\end{center}
\end{figure}
 % ------------- %
To this aim we plot in the left part of Fig.~\ref{critical3d} the lowest eigenvalue of the cut-off operator as the function of $p$ and the coupling constant. The right part shows the zero-energy cut of the surface in which the shaded region indicates the part of the $(\lambda,p)$ plane where the lowest eigenvalue of the cut-off operator is positive, as compared to $\lambda_\mathrm{crit}=\gamma_p$. The two curves meet at $p\approx  20.392$ corresponding to $\lambda_\mathrm{crit}\approx 1.563$. Up to this value, it is thus reasonable to expect that a negative eigenvalue exists. For higher values of $p$ the numerical accuracy is a demanding problem, we nevertheless conjecture that at least the Dirichlet region operator, $p=\infty$, is positive. Fig.~\ref{critical3d} also provides an idea of how the spectral threshold of $L_p(\lambda)$ depends on the coupling constant.

%%%%%%%%%%%%%%%%%%%%%%%%%%%%%%%%%%%%%%%%%%%%%%%%%%%%%%%%%%%%%
\section{Subcritical case, eigenvalue estimates}\label{s: LT}
\setcounter{equation}{0}

Let us finally pass to the subcritical case, $\lambda<\gamma_p$. According to \cite[Thm.~2.1]{EB12} the operator $L_p(\lambda)$ has in this case a purely discrete spectrum. In the mentioned paper a crude bound on eigenvalue sums was established for small values of the coupling constant $\lambda$. We are going derive now a substantially stronger result, an estimate on eigenvalue moments valid for any $\lambda<\gamma_p$. More specifically, let $\mu_1<\mu_2\le\mu_3\le\cdots$ be the set of ordered eigenvalues of (\ref{operator}); we are looking for bounds of the quantities $\sum_{j=1}^\infty (\Lambda-\mu_j)_+^\sigma$ for fixed numbers $\Lambda$ and $\sigma$. This is the contents of the following theorem.
 % ------------- %
 \begin{theorem} \label{thm:subcrit-est}
Let $\lambda<\gamma_p$, then for any $\Lambda\ge0$ and $\sigma\ge3/2$ the following trace inequality holds,
 % ------------- %
\begin{eqnarray}\label{LTbound}
\lefteqn{\hspace{-4em} \mathrm{tr} \left(\Lambda-L_p(\lambda)\right)_+^\sigma} \\&& \nonumber \hspace{-3em} \le C_{p,\sigma}\frac{\left(\Lambda+1\right)^{\sigma+(p+1)/p}}{(\gamma_p-\lambda)^{\sigma+(p+1)/p}}
\left(\left|\ln\left(\frac{\Lambda+1}{\gamma_p-\lambda}\right)\right|+1\right)+C_{p, \sigma}\, C_\lambda^2\left(\Lambda+C_\lambda^{2p/(p+2)}\right)^{\sigma+1}\!,
\end{eqnarray}
 % ------------- %
where the constant $C_{p, \sigma}$ depends on $p$ and $\sigma$ only and
 % ------------- %
$$
C_\lambda=\max\left\{\frac{1}{(\gamma_p-\lambda)^{(p+2)/(p(p+1))}},\,\frac{1}{(\gamma_p-\lambda)^{(p+2)^2/(4p(p+1))}}\right\}\,.
$$
 % ------------- %
\end{theorem}
 % ------------- %
\begin{proof}[Proof]
By the minimax principle it is sufficient to estimate $L_p$ from below by a self-adjoint operator with a purely discrete spectrum for which the moments in question can be calculated. To construct such a lower bound we again employ a bracketing, imposing additional Neumann conditions at the rectagles $G_n,\: \widetilde{G}_n$, and $Q_n,\:\widetilde{Q}_n$ introduced in the proof of Theorem~\ref{thm:subcrit}. The sequence $\{\alpha_n\}_{n=1}^\infty$ is monotonically increasing by construction; we assume again that $\alpha_n\to\infty$ and that the rectangles get asymptotically thinner according to (\ref{Decay}), i.e.
 % ------------- %
$$
\alpha_{n+1}^{2p/(p+2)}-\alpha_n^{2p/(p+2)}\to 0\quad\text{as}\quad n\to\infty\,.
$$
 % ------------- %
Then, as before, we obtain a direct sum of operators $L^{(1)}_{n, p}, \,\widetilde{L}^{(1)}_{n, p},\,\,L^{(2)}_{n, p}, \,\,\widetilde{L}^{(2)}_{n, p}$ and $L_p^0$. We are going to find the eigenvalue momentum estimates for those.

Let us start from  $L_{n, p}^{(1)},\,n=1, 2,\ldots$. We again find a lower bound using the operator $H^{(1)}_n$  given by (\ref{Hnp1}), the spectrum of which is the sum of two one-dimensional operators. Since the spectrum of one-dimensional Neumann operator $-\frac{\mathrm{d}^2}{\mathrm{d}y^2}$ on the interval $(\alpha_n, \alpha_{n+1})$ is discrete and simple with the eigenvalues $\left\{\frac{\pi^2 k^2}{(\alpha_{n+1} -\alpha_n)^2}\right\}_{k=0}^\infty$, the problem reduces to analysis of the operator $h^{(1)}_{n, p}=-\frac{\mathrm{d}^2}{\mathrm{d}x^2}+\alpha_n^p |x|^p-\lambda (x^2+\alpha_{n+1}^2)^{p/(p+2)}$ on $L^2(-\alpha_{n+1}, \alpha_{n+1})$ which is unitarily equivalent to (\ref{un.eq.}).

To proceed we put $\kappa:=\frac{\gamma_p-\lambda}{2(\gamma_p+\lambda+2)}$ and assume that the edge coordinates satisfy
 % -------------%
\begin{eqnarray}\label{decay1}
\alpha_1\ge\left(\frac{\lambda}{\kappa}\right)^{(p+2)^2/(4p(p+1))}\,,\\
\label{decay2}\alpha_{n+1}^{2p/(p+2)}-\alpha_n^{2p/(p+2)}<\frac{\kappa}{\lambda}\,.
\end{eqnarray}
 % ------------- %
Using then the fact that $(a+b)^q\le a^q+b^q$ holds for any positive numbers $a,\,b$ and $q<1$, in combination with (\ref{decay1}), we arrive at the inequalities
 % ------------- %
\begin{eqnarray*}
\lefteqn{\frac{\lambda}{\alpha_n^{2p/(p+2)}}\left(\left(\frac{x^2}{\alpha_n^{2p/(p+2)}}+\alpha_{n+1}^2\right)^{p/(p+2)}-\alpha_{n+1}^{2p/(p+2)}\right)} \\ && \le\frac{\lambda}{\alpha_n^{2p/(p+2)}}\left(\frac{x^2}{\alpha_n^{2p/(p+2)}}\right)^{p/(p+2)} \le\kappa (|x|^p+1)\,.
\end{eqnarray*}
 % ------------- %
Next, by virtue of (\ref{decay2}) and above estimate, we have
 % ------------- %
\begin{eqnarray}
\lefteqn{\nonumber h_{n,p}^{(2)}(\lambda)=\alpha_n^{2p/(p+2)}\,\biggl(-\frac{\mathrm{d}^2}{\mathrm{d}x^2}+|x|^p
-\frac{\lambda\,\alpha_{n+1}^{2p/(p+2)}}{\alpha_n^{2p/(p+2)}}} \\ && \nonumber \quad -\frac{\lambda}{\alpha_n^{2p/(p+2)}}\biggl(\left(\frac{x^2}{\alpha_n^{2p/(p+2)}}+\alpha_{n+1}^2\right)^{p/(p+2)}
-\alpha_{n+1}^{2p/(p+2)}\biggr)\biggr)\\ && \nonumber \ge\alpha_n^{2p/(p+2)} \left(-\frac{\mathrm{d}^2}{\mathrm{d}x^2}+(1-\kappa)|x|^p-\kappa-\frac{\lambda \alpha_{n+1}^{2p/(p+2)}}{\alpha_n^{2p/(p+2)}}\right)
\\ && \nonumber =\alpha_n^{2p/(p+2)} \left(-\frac{\mathrm{d}^2}{\mathrm{d}x^2}+(1-\kappa)|x|^p-\kappa-\frac{\lambda \left(\alpha_{n+1}^{2p/(p+2)}-\alpha_n^{2p/(p+2)}\right)}{\alpha_n^{2p/(p+2)}}-\lambda\right)
\\ && \label{as.}\ge (1-\kappa) \alpha_n^{2p/(p+2)}\,\left(-\frac{\mathrm{d}^2}{\mathrm{d}x^2}+|x|^p-\lambda^\prime-\kappa^\prime\right)-\kappa\,,
\end{eqnarray}
 % ------------- %
where $\kappa^\prime:=\frac{\kappa}{1-\kappa},\,\lambda^\prime:=\frac{\lambda}{1-\kappa}$, and the corresponding Neumann (an)harmonic oscillator is defined on the interval
 % ------------- %
\begin{equation}\label{interval}
(-\alpha_{n+1}\,\alpha_n^{p/(p+2)}, \alpha_{n+1}\,\alpha_n^{p/(p+2)})\,.
\end{equation}
 % ------------- %
It follows from Lemma~\ref{lemma ground} that if the interval (\ref{interval}) is large enough, which can be achieved by choosing
 % ------------- %
\begin{equation}\label{new assump.}
\alpha_2 \alpha_1^{p/(p+2)}>\alpha_1^{2(p+1)/(p+2)}>K_{0, p}
\end{equation}
 % ------------- %
with a large enough $K_{0, p}$, we have the estimate
 % ------------- %
\begin{equation}\label{inf}
h_{n, p}^{(2)}\ge(1-\kappa)\alpha_n^{2p/(p+2)}\left(\gamma_p-\frac{1}{\alpha_{n+1}^{p/2} \alpha_n^{p^2/(2(p+2))}}-\lambda^\prime-\kappa^\prime\right)-\kappa\,.
\end{equation}
 % ------------- %
Our aim is now to show that by choosing a suitable sequence $\{\alpha_n\}_{n=1}^\infty$ we can achieve  that for any $n\ge1$ the following estimate holds,
 % ------------- %
\begin{equation}\label{inf1}
\inf \sigma\left(h_{n, p}^{(2)}\right)\ge(1-\kappa) \alpha_n^{2p/(p+2)}\frac{(\gamma_p-\lambda)}{2}-\kappa\,.
\end{equation}
 % ------------- %
This is ensured, for instance, if
 % ------------- %
\begin{equation}\label{alpha_1guarant.}
\alpha_1\ge \left(\frac{2(1-\kappa)}{\gamma_p-\lambda-\kappa (\gamma_p+\lambda+2)}\right)^{(p+2)/(p(p+1))}\,.
\end{equation}
 % ------------- %
Combining (\ref{decay1}), (\ref{new assump.}) and  (\ref{alpha_1guarant.}) we thus choose
 % ------------- %
\begin{eqnarray}\label{alpha2}
\lefteqn{\hspace{-5.5em} \alpha_1=1} \\ && \nonumber \hspace{-5em} +\left[\max\left\{\!K_{0, p}^{(p+2)/(2(p+1))},
\left(\frac{2(1-\kappa)}{\gamma_p-\lambda-\kappa(\gamma_p+\lambda+2)}\right)^{(p+2)/(p(p+1))}\!\!\!,
\left(\frac{\lambda}{\kappa}\right)^{(p+2)^2/(4p(p+1))}\right\}\right]\!,
\end{eqnarray}
 % ------------- %
where $[\cdot]$ means the entire part. Let us now return to the eigenvalue momentum estimates. One has
 % ------------- %
\begin{equation}\label{eigenvalue momentum}
\mathrm{tr} \left(\Lambda-h_{n, p}^{(2)}\right)_+^\sigma
=\inf\,\sigma\left(h_{n, p}^{(2)}-\Lambda\right)_-^\sigma+\mathrm{tr}^{\prime}\left(h_{n, p}^{(2)}-\Lambda\right)_-^\sigma\,,
\end{equation}
 % ------------- %
where $\mathrm{tr}^\prime$ the summation which yields the corresponding eigenvalue moment in which the ground state is not taken into account. Using next inequalities (\ref{as.}), (\ref{inf1}), in combination with version of Lieb-Thiring inequality suitable for our purpose \cite{M15}), we infer from (\ref{eigenvalue momentum}) that for any positive $\Lambda,\,\sigma\ge3/2$ and $n\ge1$ one has
 % ------------- %
\begin{eqnarray}
\lefteqn{\hspace{-5.5em}\nonumber \mathrm{tr} \left(\Lambda-h_{n, p}^{(2)}\right)_+^\sigma\le\left(\Lambda+\kappa\right)^\sigma} \\ && \hspace{-5em}\nonumber +(1-\kappa)^\sigma\,\alpha_n^{2p\sigma/(p+2)}\,L_{\sigma, 1}^{\mathrm{cl}} \int_{-\alpha_{n+1} \alpha_n^{2p/(p+2)}}^{ \alpha_{n+1} \alpha_n^{2p/(p+2)}}\left(\frac{\Lambda+\kappa}{(1-\kappa)\,\alpha_n^{2p/(p+2)}}
-|x|^p+\lambda^\prime+\kappa^\prime\right)_+^{\sigma+1/2}\,\mathrm{d}x \\ && \nonumber \hspace{-5em}
\le\left(\Lambda+\kappa\right)^\sigma+(1-\kappa)^\sigma\,\alpha_n^{2p\sigma/(p+2)}\,L_{\sigma, 1}^{\mathrm{cl}} \int_{\mathbb{R}}\left(\frac{\Lambda+\kappa}{(1-\kappa)\,\alpha_n^{2p/(p+2)}}-|x|^p+\lambda^\prime+\kappa^\prime\right)_+^{\sigma+1/2} \mathrm{d}x
\\ && \hspace{-5em} \label{L.T.h_n} \le\left(\Lambda+\kappa\right)^\sigma+2 \alpha_n^{2p\sigma/(p+2)}\,L_{\sigma, 1}^{\mathrm{cl}}\left(\frac{\Lambda+\kappa}{(1-\kappa) \alpha_n^{2p/(p+2)}}+\lambda^\prime+\kappa^\prime\right)^{\sigma+(p+2)/(2p)}\,.
\end{eqnarray}
 % ------------- %
We further restrict the choice of the sequence $\{\alpha_n\}_{n=1}^\infty$ demanding
 % ------------- %
 \begin{equation}\label{decay*}
 \alpha_{n+1}-\alpha_n<\pi\left(\Lambda-\inf\,\sigma(h_{n, p}^{(2)})(\lambda)\right)_+^{-1/2}\,;
 \end{equation}
 % ------------- %
this allows us to write the following estimate
 % ------------- %
\begin{eqnarray}\label{Ln}
\lefteqn{\mathrm{tr} \left(\Lambda-\bigoplus_{n=1}^\infty L_{n, p}^{(1)}\right)_+^\sigma\le\mathrm{tr} \left(\Lambda-\bigoplus_{n=1}^\infty H_n^{(1)}\right)_+^\sigma} \\ && \nonumber
\le\sum_{n=1}^\infty\sum_{k=0}^\infty\mathrm{tr} \left(\Lambda-\frac{\pi^2 k^2}{(\alpha_{n+1}-\alpha_n)^2}-h_{n, p}^{(2)}\right)_+^\sigma \le\sum_{n=1}^\infty\mathrm{tr} \left(\Lambda-h_{n, p}^{(2)}\right)_+^\sigma\,.
\end{eqnarray}
 % ------------- %
Using next the fact that $\inf \sigma\left(L_{n, p}^{(1)}\right)\ge \inf \sigma\left(h_{n, p}^{(2)}\right)$ in combination with estimates  (\ref{inf1}), (\ref{L.T.h_n}), and (\ref{Ln}) one gets
 % ------------- %
\begin{eqnarray}
\lefteqn{\hspace{-5.5em} \label{tr.est.} \mathrm{tr} \left(\Lambda-\bigoplus_{n=1}^\infty L_{n,p}^{(1)}\right)_+^\sigma \le\sum_{(\gamma_p-\lambda)\alpha_n^{2p/(p+2)}<\frac{2(\Lambda+\kappa)}{1-\kappa}}\left(\Lambda+\kappa\right)^\sigma}
\\ && \hspace{-5em} \nonumber +2 L_{\sigma, 1}^{\mathrm{cl}}\sum_{(\gamma_p-\lambda) \alpha_n^{2p/(p+2)}<\frac{2(\Lambda+\kappa)}{1-\kappa}}\alpha_n^{2p\sigma/(p+2)}\,\biggl(\frac{\Lambda+\kappa}{(1-\kappa) \alpha_n^{2p/(p+2)}}+\lambda^\prime+\kappa^\prime\biggr)^{\sigma+(p+2)/(2p)} \\ && \hspace{-5em} \nonumber
\le\sum_{(\gamma_p-\lambda)\alpha_n^{2p/(p+2)}<\frac{2(\Lambda+\kappa)}{1-\kappa}}\left(\Lambda+\kappa\right)^\sigma
\\ && \hspace{-5em} \nonumber +\frac{2 L_{\sigma, 1}^{\mathrm{cl}}}{(1-\kappa)^{\sigma+(p+2)/(2p)}}\,\sum_{(\gamma_p-\lambda) \alpha_n^{2p/(p+2)}<\frac{2(\Lambda+\kappa)}{1-\kappa}}\alpha_n^{2p\sigma/(p+2)}\,
\left(\frac{\Lambda+\kappa}{\alpha_n^{2p/(p+2)}}+\lambda+\kappa\right)^{\sigma+(p+2)/(2p)} \\ && \hspace{-5em} \nonumber \le\left(\Lambda+\kappa\right)^\sigma \#\left\{\alpha_n<\left(\frac{2(\Lambda+\kappa)}{(1-\kappa)(\gamma_p-\lambda)}\right)^{(p+2)/(2p)}\right\}
\\ && \hspace{-5em} \nonumber +\frac{2L_{\sigma,1}^{\mathrm{cl}}\,(\lambda+1+\kappa)^{\sigma+(p+2)/(2p)}}{(1-\kappa)^{\sigma+(p+2)/(2p)}}\,
(\Lambda+\kappa)^{\sigma+(p+2)/(2p)}\,\sum_{\alpha_n<(\Lambda+\kappa)^{(p+2)/(2p)}}\frac{1}{\alpha_n}
\\ && \hspace{-5em} \nonumber +\frac{2 L_{\sigma, 1}^{\mathrm{cl}}\,(\lambda+1+\kappa)^{\sigma+(p+2)/(2p)}}{(1-\kappa)^{\sigma+(p+2)/(2p)}}\,
\sum_{(\Lambda+\kappa)^{(p+2)/(2p)}<\alpha_n<\frac{1}{(\gamma_p-\lambda)^{(p+2)/(2p)}}\left(\!
 \frac{2(\Lambda+\kappa)}{1-\kappa}\right)^{(p+2)/(2p)}}\alpha_n^{2p\sigma/(p+2)}\,,
\end{eqnarray}
 % ------------- %
where $\#\{\cdot\}$ means the cardinality of the corresponding set.

Using the same technique one obtains estimates for operators $\widetilde{L}_{n, p}^{(1)},\,L_{n, p}^{(2)},\,\widetilde{L}_{n, p}^{(2)}$ analogous to (\ref{tr.est.}). Finally, the operator $L_p^0$ can be estimated from below by
 % ------------- %
$$
H_{0, p}=-\frac{\partial^2}{\partial x^2}-\frac{\partial^2}{\partial y^2}-2^{p/(p+2)}\lambda \alpha_1^{2p/(p+2)}\quad\text{on}\quad G_0
$$
 % ------------- %
with Neumann  conditions at the boundary $\partial G_0$. The spectrum of $H_{0, p}$ is
 % ------------- %
$$
\left\{\frac{\pi^2 k^2}{4\alpha_1^2}+\frac{\pi^2 m^2}{4\alpha_1^2}-2^{p/(p+2)}\lambda \alpha_1^{2p/(p+2)}\right\}_{k, m=0}^\infty\,,
$$
 % ------------- %
and therefore
 % ------------- %
\begin{eqnarray}
\lefteqn{ \hspace{-5.5em} \nonumber \mathrm{tr}\left(\Lambda-H_{0, p}\right)_+^\sigma\le\sum_{k, m=0}^\infty\left(\Lambda+2^{p/(p+2)}\lambda \alpha_1^{2p/(p+2)}-\frac{\pi^2 k^2}{4\alpha_1^2}-\frac{\pi^2 m^2}{4\alpha_1^2}\right)_+^\sigma} \\ && \nonumber \hspace{-5em} \le\left(\Lambda+2^{p/(p+2)}\lambda \alpha_1^{2p/(p+2)}\right)^\sigma \\ && \nonumber \hspace{-5em} \times \sum_{k=0}^{2\alpha_1\sqrt{\Lambda+2^{p/(p+2)}\lambda \alpha_1^{2p/(p+2)}}/\pi}\left(\frac{2 \alpha_1}{\pi}\biggl(\Lambda+2^{p/(p+2)}\lambda \alpha_1^{2p/(p+2)}-\frac{\pi^2 k^2}{4\alpha_1^2}\right)^{1/2}+1\biggr) \\ && \hspace{-5em} \label{L_0} \le\left(\frac{2\alpha_1}{\pi}\left(\Lambda+2^{p/(p+2)}\lambda \alpha_1^{2p/(p+2)}\right)^{1/2}+1\right)^2\,\left(\Lambda+
2^{p/(p+2)}\lambda \alpha_1^{2p/(p+2)}\right)^\sigma\,.
\end{eqnarray}
 % ------------- %
Consider now $\alpha_1$ is defined in (\ref{alpha2}) and to any $\nu=\alpha_1, \alpha_1+1, \alpha_1+2, \ldots\,$ define a finite sequence of numbers by $\beta_k(\nu)=\nu+\frac{k}{\left[\nu^{p/(p+2)} \ln\nu\right]},\,k=0, 1,\ldots,\left[\nu^{p/(p+2)}\,\ln\nu\right]-1$. This allows us to construct a sequence $\{\alpha_n\}_{n=1}^\infty$ of the rectangle edge coordinates using the following prescription: the first term is given by (\ref{alpha2}) and the further ones are $\alpha_2=\beta_1(\alpha_1),\,\ldots,\alpha_{\left[\alpha_1^{p/(p+2)}\,\ln \alpha_1\right]} =\beta_{\left[\alpha_1^{p/(p+2)}(\alpha_1)\,\ln \alpha_1\right]-1},\,\,\alpha_{\left[\alpha_1^{p/(p+2)}\,\ln \alpha_1\right]+1}=\beta_0(\alpha_1+1),  \ldots\,$, etc., where $[\cdot]$ as usual denotes the entire part. With this choice of $\{\alpha_n\}_{n=1}^\infty$, one can check that the right-hand side of (\ref{tr.est.}) is not larger than
 % ------------- %
\begin{eqnarray}
\nonumber C_{p, \sigma}\biggl(\frac{\left(\Lambda+\kappa\right)^\sigma}{(\gamma_p-\lambda)^\sigma}\,\max\left\{0,\:
\frac{(\Lambda+\kappa)^{(p+1)/p}}{(\gamma_p-\lambda)^{(p+1)/p}}\,\ln\left(\frac{2(\Lambda+\kappa)}{(1-\kappa)(\gamma_p-\lambda)}\right)\right\}
\\\label{final.in.} \quad +\left(\Lambda+\kappa\right)^{\sigma+1/2+1/p}\,\max\left\{0,\:
\left(\Lambda+\kappa\right)^{1/2}\,\ln\left(\Lambda+\kappa\right)\right\}\biggr)
\end{eqnarray}
 % ------------- %
with a constant depending on $p$ and $\sigma$ only. On the other hand, the right-hand side of (\ref{L_0}) is not larger than
 % ------------- %
$$
\tilde{C}_{p, \sigma}\alpha_1^2\left(\Lambda+\alpha_1^{2p/(p+2)}\right)^{\sigma+1}
$$
 % ------------- %
with another constant $\tilde{C}_{p, \sigma}$.  In this way the theorem is established.

\end{proof}

\section*{Acknowledgments}
We are obliged to Ari Laptev for a useful discussion. The research has been supported by the Czech Science Foundation (GA\v{C}R) within the project 14-06818S. D.B. acknowledges the support of the University of Ostrava and the project ``Support of Research in the Moravian-Silesian Region 2013''.
The research of A.K. is supported by the German Research Foundation through  CRC 1173 ``Wave phenomena: analysis and numerics''.

\end{document}